\documentclass{article}

\usepackage{natbib}
\usepackage{imakeidx} 

\usepackage{multirow}
\usepackage{units}
\usepackage{amsmath}
\usepackage{amsthm}
\usepackage{amssymb}
\usepackage[all]{xy}

\theoremstyle{plain}
\newtheorem{thm}{\protect\theoremname}[section]
\theoremstyle{definition}
\newtheorem{defn}[thm]{\protect\definitionname}
\theoremstyle{plain}
\newtheorem{lem}[thm]{\protect\lemmaname}

\providecommand{\definitionname}{Definition}
\providecommand{\lemmaname}{Lemma}
\providecommand{\theoremname}{Theorem}

\date{}

\begin{document}

\title{Probabilistic Contextuality\\in EPR/Bohm-type Systems with Signaling Allowed}

\author{Janne V. Kujala  \\
	University of Jyv\"askyl\"a \\
	\and 
	Ehtibar N. Dzhafarov \\
	Purdue University \\
	}

\maketitle

\begin{abstract}
In this chapter, we review a principled way of defining and measuring
contextuality in systems with deterministic inputs and random outputs,
recently proposed and developed in \citep{KujalaDzhafarovLarsson2015,DKL2015FooP}.
We illustrate it on systems with two binary inputs and two binary
random outputs, the prominent example being the system of two entangled
spin-half particles with each particle's spins (random outputs) being
measured along one of two directions (inputs). It is traditional to
say that such a system exhibits contextuality when it violates Bell-type
inequalities. Derivations of Bell-type inequalities, however, are
based on the assumption of no-signaling, more generally referred to
as marginal selectivity: the distributions of outputs (spins) in Alice's
particle do not depend on the inputs (directions chosen) for Bob's
particle. In many applications this assumption is not satisfied, so
that instead of contextuality one has to speak of direct cross-influences,
e.g., of Bob's settings on Alice's spin distributions. While in quantum
physics direct cross-influences can sometimes be prevented (e.g.,
by space-like separation of the two particles), in other applications,
especially in behavioral and social systems, marginal selectivity
almost never holds. It is unsatisfying that the highly meaningful
notion of contextuality is made inapplicable by even slightest violations
of marginal selectivity. The new approach rectifies this: it allows
one to define and measure contextuality on top of direct cross-influences,
irrespective of whether marginal selectivity (no-signaling condition)
holds. For systems with two binary inputs and two binary random outputs,
contextuality means violation of the classical CHSH inequalities in
which the upper bound 2 is replaced with 2$\left(1+\Delta_{0}\right)$,
where $\Delta_{0}$ is a measure of deviation from marginal
selectivity.
\end{abstract}


\section{\label{sec:Introduction}Introduction}

In the foundations of quantum physics the notion of contextuality
can be formulated in purely probabilistic terms within the framework
of the Kolmogorovian probability theory \citep*{Larsson2002,Khrennikov2008_Bell-Boole,Khrennikov2008_EPR-Bohm,DK2013PLoS,DK2014LNCSQualified,DK2014PLOSconditionalization,DK2014FOOP,DK2014Advances,DK2014Scripta,DKL2015FooP}.
The notion applies to any system of random variables recorded under
different (mutually incompatible) conditions. Contextuality means
that these random variables cannot be ``sewn together'' into a single
system of jointly distributed random variables if one assumes that
all or some of them preserve their identity across different conditions.
Within the Kolmogorovian framework the existence of this single joint
distribution is equivalent to the presentability of all random variables
involved as functions of one and the same (``hidden'') random variable
\citep{SuppesZanotti1981,Fine1982b,DzhafarovKujala2010}.

In spite of its long history (dating from Specker's \citeyearpar{Specker1960}
example with three boxes, contextuality does not have a standard definition
\citep{KochenSpecker1967,Laudisa1997,Spekkens2008,Kirchmair2009,Badziag2009,Khrennikov2009,Cabello2013PRL},
and is often confounded with such notions as nonlocality and lack
of realism (the notions we will not get into in this chapter). All
authors who use this term in quantum theory, however, agree on the
possibility of detecting contextuality in the spins of entangled particles
by violations of Bell-type inequalities \citep{Fine1982b,Bell1964,ClauserHorneShimonyHolt1969}.
Many other tests have been developed for systems of random variables
in and outside quantum physics, notably in psychology \citep{KujalaDzhafarov2008b,DzhafarovKujala2012a,DzhafarovKujala2012b,DzhafarovKujala2013ProcAMS}.
All of these tests are necessary (sometimes also sufficient) conditions
for non-contextuality, because of which all of them presuppose or
are directly making use of the condition known in psychology as marginal
selectivity \citep{TownsendSchweickert1989,Dzhafarov2003c} and in
quantum physics as no-signaling \citep{Cereceda2000,Masanes2006,Oas2014}.
In this chapter we use the first term, as more general and purely
probabilistic (see Section \ref{sec:Consequences-of-the}).\footnote{Within our most recent publications developing the theory, this property
is also referred to by the technical name \index{consistently connected}\emph{consistent connectedness}.} If marginal selectivity is violated, no ``sewing together'' of
the kind mentioned above is possible.

The problem associated with this fact is that in some cases (including
all cases known to us in psychology) violations of marginal selectivity
can be readily attributed to the lack of selectivity in the dependence
of random variables on various components of the conditions under
which they are recorded. If a person is asked to judge brightness
and size of a visually presented object, it is not difficult to construct
a model in which the judgment of brightness is directly influenced
by physical intensity and also directly influenced by object's physical
size. In the EPR/Bohm paradigm, if the two measurements of spins in
entangled particles are separated by a time-like interval, the spatial
axis chosen by Bob (for one of the particles) can in principle initiate
a process that will directly influence the spin recorded by Alice
(for another particle). We will refer to the dependence of an output
distribution on the ``wrong'' input as a direct cross-influence.
The Bell-type inequalities (e.g., in the CHSH form, \citealp{ClauserHorneShimonyHolt1969})
cannot be derived under direct cross-influences, and whether or not
they are violated therefore becomes irrelevant.

It seems strange and intellectually unsatisfying, however, that we
can detect contextuality when marginal selectivity holds precisely,
but we cannot speak of contextuality at all when it is violated, however
slightly. In this chapter we review (in the context of systems with
binary inputs and binary random variables as outputs) a recently proposed
definition and measure of contextuality \citep{KujalaDzhafarovLarsson2015,DKL2015FooP}
that overcome this difficulty: even in the presence of direct cross-influences
(say, from Bob's setting to Alice's measurements and vice versa) we
can detect the presence and compute the degree of contextual influences
``on top of'' the direct cross-influences. The theory can be generalized
to arbitrary systems with deterministic inputs and random outputs,
but we do not attempt to present it here. We have made an effort to
keep the presentation on a very nontechnical level. This level would
be difficult to maintain in a more systematical or more general presentation.

\section{\label{sec:The-System-The}The System $\left(\alpha,\beta,A,B\right)$}

Consider a system with two binary inputs, $\alpha,\beta$, and two
outputs that are binary random variables, $A,B$. Alice chooses the
value of $\alpha$ to be either $\alpha_{1}$ or $\alpha_{2}$, and
she records the corresponding value of $A$ as either $+1$ or $-1$.
Bob chooses the value of $\beta$ to be either $\beta_{1}$ or $\beta_{2}$,
and he records the value of $B$ as either $+1$ or $-1$. Alice and
Bob do this repeatedly in successive trials, so that each input choice
and output recording by Alice is paired with an input choice and output
recording by Bob. They send their paired choices of inputs and recordings
of the outputs to Charlie, who creates four tables of joint distributions:
for every $i\in\left\{ 1,2\right\} $ and $j\in\left\{ 1,2\right\} $,
the distribution is\begin{equation}%
\begin{tabular}{|c|cc|c|}
\cline{1-3} 
$\phi=(\alpha_{i},\beta_{j})$ & $B_{ij}=+1$ & $B_{ij}=-1$ & \multicolumn{1}{c}{}\tabularnewline
\hline 
$A_{ij}=+1$ & $\Pr\left[A_{ij}=1,B_{ij}=1\right]$ & $\ldots$ & $\Pr\left[A_{ij}=1\right]$\tabularnewline
$A_{ij}=-1$ & $\ldots$ & $\ldots$ & $\ldots$\tabularnewline
\hline 
\multicolumn{1}{c|}{} & $\Pr\left[B_{ij}=1\right]$ & $\ldots$ & \multicolumn{1}{c}{}\tabularnewline
\cline{2-3} 
\end{tabular}\end{equation}Charlie knows that the only variables that can possibly
influence $A$ are $\alpha$ and $\beta$, so he labels $A$ recorded
under conditions $\phi=\left(\alpha_{i},\beta_{j}\right)$ as $A_{ij}$,
allowing thereby $A_{ij}$ to have up to four different distributions.
Each of these distributions can be represented by $\Pr\left[A_{ij}=1\right]$,
or equivalently by the expected value $\left\langle A_{ij}\right\rangle =2\Pr\left[A_{ij}=1\right]-1$.
The notation $B_{ij}$ for Bob, and the values $\Pr\left[B_{ij}=1\right]$
and $\left\langle B_{ij}\right\rangle $ are analogous.

Charlie thus deals with eight random variables, 
\begin{equation}
A_{11},B_{11},A_{12},B_{12},A_{21},B_{21},A_{22},B_{22}.\label{eq:the 8}
\end{equation}
With respect to the joint distribution of $A_{ij}$ and $B_{ij}$,
their individual distributions are referred to as marginal. The joint
distribution for $\left(A_{ij},B_{ij}\right)$ is uniquely determined
by the two marginal probabilities and the joint probability $\Pr\left[A_{ij}=+1\textnormal{ and }B_{ij}=+1\right]$.
Equivalently, it is determined by the two expected values $\left\langle A_{ij}\right\rangle ,\left\langle B_{ij}\right\rangle $
and the product expected value 
\begin{equation}
\left\langle A_{ij}B_{ij}\right\rangle =\Pr\left[A_{ij}=B_{ij}\right]-\Pr\left[A_{ij}\not=B_{ij}\right].
\end{equation}

\section{\label{sec:Selectivity-of-influences}Selectivity of influences and
marginal selectivity}

\index{selective influence}Let us assume that Charlie, based on some theory, expects that the
dependence of $A,B$ on $\alpha,\beta$ is selective: Bob's choice
of $\beta$ value does not influence Alice's $A$ and vice versa:
\begin{equation}
\begin{array}{c}
\xymatrix{\alpha\ar[d] & \beta\ar[d]\\
A & B
}
\end{array}\label{eq:diagram selective}
\end{equation}
This means that $A_{i1}$ and $A_{i2}$ are one and the same random
variable for every $i\in\left\{ 1,2\right\} $, and so are $B_{1j}$
and $B_{2j}$ for every $j\in\left\{ 1,2\right\} $. Charlie can therefore
relabel $A_{ij}$ into $A_{i}$ and $B_{ij}$ into $B_{j}$. But he
can also approach this in a more cautious way. He can retain the double
indexation and ask the following question: given the eight random
variables in (\ref{eq:the 8}) of which we know the expectations 
\begin{equation}
\left(\left\langle A_{ij}B_{ij}\right\rangle ,\left\langle A_{ij}\right\rangle ,\left\langle B_{ij}\right\rangle \right),\;i,j\in\left\{ 1,2\right\} ,\label{eq:expectations observed}
\end{equation}
can we impose a joint distribution on these eight random variables\footnote{To impose a joint distribution on (\ref{eq:the 8}) means to create
a vector of jointly distributed $A'_{11}$, $B'_{11}$, $\ldots$,
$A'_{22}$, $B'_{22}$ called a \index{coupling}coupling for (\ref{eq:the 8}), such
that the pairs $\left(A'_{ij},B'_{ij}\right)$ have the same distributions
as $\left(A_{ij},B_{ij}\right)$ for all $i,j\in\left\{ 1,2\right\} $.
No other subset of (\ref{eq:the 8}) has a joint distribution. In
this chapter we conveniently confuse random variables and their primed
counterparts. See \citep{DK2013PLoS,DK2014LNCSQualified,DK2014PLOSconditionalization,DK2014FOOP,DK2014Advances,DzhafarovKujala2010,DzhafarovKujala_handbook}
for detailed discussions. } such that 
\begin{equation}
\begin{array}{cc}
\Pr\left[A_{i1}\not=A_{i2}\right]=0 & \textnormal{ for }i\in\left\{ 1,2\right\} \\
\Pr\left[B_{1j}\not=B_{2j}\right]=0 & \textnormal{ for }j\in\left\{ 1,2\right\} 
\end{array}?\label{eq:identity connection}
\end{equation}
If the answer is affirmative, then the situation is equivalent to
the existence of a joint distribution of the single-indexed $A_{1},B_{1},A_{2},B_{2}$
such that 
\begin{equation}
\left(\left\langle A_{i}B_{j}\right\rangle ,\left\langle A_{i}\right\rangle ,\left\langle B_{j}\right\rangle \right)=\left(\left\langle A_{ij}B_{ij}\right\rangle ,\left\langle A_{ij}\right\rangle ,\left\langle B_{ij}\right\rangle \right),\;i,j\in\left\{ 1,2\right\} .
\end{equation}

However, and this is the reason we call Charlie's approach cautious,
the answer does not have to be affirmative. One situation that precludes
this is if the following equalities are violated at least for one
$i$ or one $j$: 
\begin{equation}
\left\langle A_{i1}\right\rangle =\left\langle A_{i2}\right\rangle ,\;\left\langle B_{1j}\right\rangle =\left\langle B_{2j}\right\rangle .\label{eq:marginal invariance}
\end{equation}\index{marginal selectivity}%
These equalities represent marginal selectivity of $A$ with respect
to changes in $\beta$ and of $B$ with respect to changes in $\alpha$.
This marginal selectivity is an obvious consequence of (\ref{eq:identity connection}).
If, e.g., $\left\langle A_{11}\right\rangle $ were different from
$\left\langle A_{12}\right\rangle $, then, as Bob changes the value
of $\beta$ from $\beta_{1}$ to $\beta_{2}$, Alice's distribution
of $A$ for one and the same choice of $\alpha=\alpha_{1}$ changes.
$A_{11}$and $A_{12}$ cannot therefore be always equal, contravening
(\ref{eq:identity connection}).

In situations like this Charlie is forced then to revise his model
(\ref{eq:diagram selective}) in favor of 
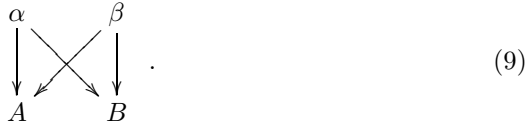
\begin{equation}
\begin{array}{c}
\xymatrix{\alpha\ar[d]\ar[dr] & \beta\ar[d]\ar[dl]\\
A & B
}
\end{array}.\label{eq:diagrams cross}
\end{equation}
This can be referred to as a model with direct cross-influences: the
distribution (hence also identity) of the outputs is allows to be
influenced by ``wrong'' inputs (``wrong'' from the point of view
of the Charlie's original theory\footnote{This is the ``subjective'', or theory-laden aspect of the notion
of contextuality: this notion acquires its meaning only in relation to some
model, in this case represented by (\ref{eq:diagram selective}),
that describes the system the way it ``ought to be'' or predicted
to be by some theory. We will not elaborate, but this accords with
our view \citep{DK2014Scripta} that while probabilities are objective,
the identities of random variables are theory-laden.}).

\section{\label{sec:Contextuality-under-marginal}Contextuality under marginal
selectivity}\index{contextuality!---, under marginal selectivity}

There is also another possibility for Charlie's question to have a
negative answer. The marginal selectivity requirement may very well
be satisfied, but the observed expectations (\ref{eq:expectations observed})
may be incompatible with the hypothesis (\ref{eq:identity connection}).
The incompatibility means that a joint distribution of the eight random
variables (\ref{eq:the 8}) that accords with both (\ref{eq:expectations observed})
and (\ref{eq:identity connection}) does not exist. This understanding
of contextuality was first utilized by \citet{Larsson2002}. It helps
to understand the essence of all Bell-type theorems. Stated in the
form convenient for our purposes, the theorem that applies to all
systems with two binary inputs and two binary random outputs \citet{Fine1982b}
says: 
\begin{thm}
[Fine, 1982]\label{thm:Fine}The observed expectations (\ref{eq:expectations observed})
are compatible with the identity connections (\ref{eq:identity connection})
if and only if marginal selectivity (\ref{eq:marginal invariance})
is satisfied for all $i,j\in\left\{ 1,2\right\} $, and 
\begin{equation}
\max_{i,j\in\left\{ 1,2\right\} }\left|\left\langle A_{11}B_{11}\right\rangle +\left\langle A_{12}B_{12}\right\rangle +\left\langle A_{21}B_{21}\right\rangle +\left\langle A_{22}B_{22}\right\rangle -2\left\langle A_{ij}B_{ij}\right\rangle \right|\leq2.\label{eq:Bell/Fine/CHSH}
\end{equation}
\end{thm}
The term \index{connection}``connections'' used in this formulation \citep{DK2013PLoS,DK2014LNCSQualified,DK2014PLOSconditionalization,DK2014FOOP}
refers to the unobservable pairs 
\begin{equation}
\left(A_{11},A_{12}\right),\left(A_{21},A_{22}\right),\left(B_{11},B_{21}\right),\left(B_{12},B_{22}\right).
\end{equation}
Their unobservable joint distributions are given by\begin{equation}%
\begin{tabular}{c|cc|c}
\cline{2-3} 
 & $A_{i2}=+1$ & $A_{i2}=-1$ & \tabularnewline
\hline 
\multicolumn{1}{|c|}{$A_{i1}=+1$} & $\Pr\left[A_{i1}=1,A_{i2}=1\right]$ & $\ldots$ & \multicolumn{1}{c|}{$\Pr\left[A_{i1}=1\right]$}\tabularnewline
\multicolumn{1}{|c|}{$A_{i1}=-1$} & $\ldots$ & $\ldots$ & \multicolumn{1}{c|}{$\ldots$}\tabularnewline
\hline 
 & $\Pr\left[A_{i2}=1\right]$ & $\ldots$ & \tabularnewline
\cline{2-3} 
\multicolumn{1}{c}{} &  & \multicolumn{1}{c}{} & \tabularnewline
\cline{2-3} 
 & $B_{2j}=+1$ & $B_{2j}=-1$ & \tabularnewline
\hline 
\multicolumn{1}{|c|}{$B_{1j}=+1$} & $\Pr\left[B_{1j}=1,B_{2j}=1\right]$ & $\ldots$ & \multicolumn{1}{c|}{$\Pr\left[B_{1j}=1\right]$}\tabularnewline
\multicolumn{1}{|c|}{$B_{1j}=-1$} & $\ldots$ & $\ldots$ & \multicolumn{1}{c|}{$\ldots$}\tabularnewline
\hline 
 & $\Pr\left[B_{2j}=1\right]$ & $\ldots$ & \tabularnewline
\cline{2-3} 
\end{tabular}\end{equation}for $i,j\in\left\{ 1,2\right\} $. If (\ref{eq:identity connection})
holds, i.e., the entries on the minor diagonals of the tables are
zero, then the connections are called the identity ones.

The compatibility of connections with the observed expectations (uniquely
defining the observed distributions) means that each of the $2^{8}$
possible combinations 
\[
A_{11}=\pm1,B_{11}=\pm1,\ldots,A_{22}=\pm1,B_{22}=\pm1
\]
is assigned a probability, so that the probabilities for all combinations
containing, say, $A_{12}=1$ and $B_{12}=-1$ sum to the observed
$\Pr\left[A_{12}=1,B_{12}=-1\right]$; and the probabilities for all
combinations containing, say, $B_{12}=1$ and $B_{22}=1$ equals the
hypothetical (unobservable) connection probability $\Pr\left[B_{12}=1,B_{22}=1\right]$.

The inequalities (\ref{eq:Bell/Fine/CHSH}), in physics referred to
as CHSH inequalities, can be violated, and they are de facto violated
if $A$ and $B$ are spins of two entangled particles under certain
choices of spatial axes ($\alpha$ and $\beta$) along which they
are measured \citep{AspectGrangierRoger1981,AspectGrangierRoger1982,Weihs1998}.
When these inequalities are violated while marginal selectivity is
satisfied, we speak of contextuality: Alice's output $A$ under her
choice of $\alpha_{1}$ does not change its distribution depending
on Bob's choice of $\beta_{1}$ or $\beta_{2}$, but $A_{11}$ and
$A_{12}$ still cannot be considered one and the same random variable
(it should not come as a surprise that different random variables
can have the same distribution).

In the diagram below the interrupted lines indicate contextual influences:
the dependence of identities of identically distributed random variables
on the ``wrong'' inputs: 
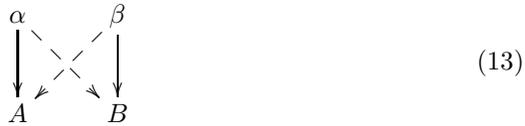
\begin{equation}
\begin{array}{c}
\xymatrix{\alpha\ar[d]\ar@{-->}[dr] & \beta\ar[d]\ar@{-->}[dl]\\
A & B
}
\end{array}
\end{equation}
When the inequalities (\ref{eq:Bell/Fine/CHSH}) are violated, a measure
of contextuality can be easily designed as follows. If (\ref{eq:identity connection})
were compatible with the observed expectations (\ref{eq:expectations observed}),
then (by definition) Charlie could construct a joint distribution
of the random variables (\ref{eq:the 8}) in which

\begin{equation}
\Delta=\Pr\left[A_{11}\not=A_{12}\right]+\Pr\left[A_{21}\not=A_{22}\right]+\Pr\left[B_{11}\not=B_{21}\right]+\Pr\left[B_{12}\not=B_{22}\right]\label{eq:C under MS}
\end{equation}
equals zero. If (\ref{eq:identity connection}) is incompatible with
(\ref{eq:expectations observed}), then this $\Delta$ cannot be zero
in any joint distribution imposed on (\ref{eq:the 8}). It is natural
therefore to adopt the following 
\begin{defn}
\label{def:contextuality under MS}Under marginal selectivity, the
degree of contextuality in a system with given observed expectations
(\ref{eq:expectations observed}) is the minimal value of $\Delta$
in (\ref{eq:C under MS}) for which a joint distribution for (\ref{eq:the 8})
exists. 
\end{defn}
As it turns out, this minimal value of $\Delta$ equals 
\begin{equation}
\Delta_{\min}=\max\left\{ 0,\Delta_{\textnormal{CHSH}}\right\} ,\label{eq:C_min}
\end{equation}
where 
\begin{equation}
  \begin{split}
    &\Delta_{\textnormal{CHSH}}=\\
    &{\textstyle\frac{1}{2}}\max_{i,j\in\left\{ 1,2\right\} }\left|\left\langle A_{11}B_{11}\right\rangle +\left\langle A_{12}B_{12}\right\rangle +\left\langle A_{21}B_{21}\right\rangle +\left\langle A_{22}B_{22}\right\rangle -2\left\langle A_{ij}B_{ij}\right\rangle \right|-1\label{eq:C^0 general}
  \end{split}
\end{equation}
is ($\nicefrac{1}{2}$ times) the violation of the CHSH inequalities.
This is a special case of the formula derived later in Theorem \ref{thm:C_min}
without the assumption of marginal selectivity.

As an example, let the observed expectations be at the Tsirelson bounds
\citep{Tsirelson1980,Landau1987}. Then $\Delta_{\min}$ is $\sqrt{2}-1$.
The largest possible value of $\Delta_{\min}$ is 1.

\section{\label{sec:Contextuality-on-top}Contextuality on top of direct cross-influences}\index{contextuality!---, on top of direct cross-influences}

The definition of contextuality given above does not work for the
situation depicted in (\ref{eq:diagrams cross}), where marginal selectivity
is not satisfied. In this case we have direct cross-influences from
``wrong'' inputs, and this precludes the possibility that $\Delta$
in (\ref{eq:C under MS}) is zero. In fact, we have the simple 
\begin{thm}
\label{thm:C_0}Given the observed expectations $\left(\left\langle A_{ij}\right\rangle ,\left\langle B_{ij}\right\rangle \right)_{i,j\in\left\{ 1,2\right\} }$,
the minimum possible value for $\Delta$ in (\ref{eq:C under MS})
is 
\begin{equation}
  \begin{split}
    \Delta_{0}=\textstyle\frac{1}{2}(&\left|\left\langle A_{11}\right\rangle -\left\langle A_{12}\right\rangle \right|+\left|\left\langle A_{21}\right\rangle -\left\langle A_{22}\right\rangle \right|+\\
    &\left|\left\langle B_{11}\right\rangle -\left\langle B_{21}\right\rangle \right|+\left|\left\langle B_{12}\right\rangle -\left\langle B_{22}\right\rangle \right|).\label{eq:C_0 general}
  \end{split}
\end{equation}
\end{thm}
\begin{proof}
We minimize $\Delta$ if we minimize separately $\Pr\left[A_{11}\not=A_{12}\right]$,
$\Pr\left[A_{21}\not=A_{22}\right]$, $\Pr\left[B_{11}\not=B_{21}\right]$,
and $\Pr\left[B_{12}\not=B_{22}\right]$. Consider, e.g., the distribution
of the connection $\left(A_{11},A_{12}\right)$:\begin{equation}%
\begin{tabular}{|c|ccc|}
\cline{2-4} 
\multicolumn{1}{c|}{} & $A_{12}=+1$ && $A_{12}=-1$\tabularnewline
\hline 
\!\!$A_{11}=+1$\!\! & \!\!$\Pr\left[A_{11}=1,A_{12}=1\right]$ & \multicolumn{2}{c|}{\quad $\Pr\left[A_{11}=1\right]-\Pr\left[A_{11}=1,A_{12}=1\right]$\!\!}\tabularnewline
\!\!$A_{11}=-1$\!\! & \multicolumn{2}{c}{\!\!$\Pr\left[A_{12}=1\right]-\Pr\left[A_{11}=1,A_{12}=1\right]$\quad } & $\ldots$\!\!\tabularnewline
\hline 
\end{tabular}\end{equation}The largest possible value for the probability $\Pr\left[A_{11}=1,A_{12}=1\right]$
is 
$$\min\left\{ \Pr\left[A_{11}=1\right],\Pr\left[A_{12}=1\right]\right\} ,$$
whence the minimum of $\Pr\left[A_{11}\not=A_{12}\right]$, which
is the sum of the entries on the minor diagonal, is $\left|\Pr\left[A_{11}=1\right]-\Pr\left[A_{12}=1\right]\right|=\frac{1}{2}\left|\left\langle A_{11}\right\rangle -\left\langle A_{12}\right\rangle \right|$. 
\end{proof}
Under the marginal selectivity we have $\Delta_{0}=0$, and we speak
of contextuality if the minimal value of $\Delta$ that is compatible
with the observed expectations (\ref{eq:expectations observed}) is
greater than $\Delta_{0}=0$. In the general case $\Delta_{0}>0$,
and we need a more general definition of contextuality. The idea is
simple. If $\Delta_{0}>0$, we have direct cross-influences (\ref{eq:diagrams cross}),
and if $\Delta=\Delta_{0}$ is compatible with the observed expectations
(\ref{eq:expectations observed}), then no contextuality is involved:
direct cross-influences is all one needs to account for the system's
behavior. If however $\Delta=\Delta_{0}$ is not compatible with the
observed expectations (\ref{eq:expectations observed}), then we can
speak of contextuality ``on top of'' the direct cross-influences.
The natural measure of the degree of contextuality then is given by 
\begin{defn}
\label{def:contextuality general}The degree of contextuality in a
system with given observed expectations (\ref{eq:expectations observed})
is $\Delta_{\min}-\Delta_{0}$, where $\Delta_{\min}$ is the minimal
value of $\Delta$ in (\ref{eq:C under MS}) for which a joint distribution
for (\ref{eq:the 8}) exists. 
\end{defn}

\section{\label{sec:General-formula-for}General formula for contextuality}

We now need to derive a formula for $\Delta_{\min}$ of which (\ref{eq:C_min})
is a special case. 
\begin{thm}
\label{thm:C_min}The minimum possible value $\Delta_{\min}$ for
$\Delta$ that is compatible with the observed expectations (\ref{eq:expectations observed})
is 
\begin{equation}
\Delta_{\min}=\max\left\{ \Delta_{0},\Delta_{\textnormal{CHSH}}\right\} ,\label{eq:C_min general}
\end{equation}
where $\Delta_{0}$ is given in (\ref{eq:C_0 general}) and $\Delta_{\textnormal{CHSH}}$
in (\ref{eq:C^0 general}).\end{thm}
\begin{proof}
By Lemma \ref{lem:2} (a computer-assisted result detailed in the next section), $\Delta$ is compatible with the observed $\left(\left\langle A_{ij}B_{ij}\right\rangle ,\left\langle A_{ij}\right\rangle ,\left\langle B_{ij}\right\rangle \right)_{i,j\in\left\{ 1,2\right\} }$
if and only if it satisfies 
\begin{align}
  \Delta &\textstyle \ge-1+\frac{1}{2}s_{1}\left(\left\langle A_{11}B_{11}\right\rangle ,\left\langle A_{12}B_{12}\right\rangle ,\left\langle A_{21}B_{21}\right\rangle ,\left\langle A_{22}B_{22}\right\rangle \right),\label{eq:S_lower_from_r}\\
  \begin{split}
    \Delta &\textstyle \ge\frac{1}{2}(
    \left|\left\langle A_{11}\right\rangle -\left\langle A_{12}\right\rangle \right|+\left|\left\langle A_{21}\right\rangle -\left\langle A_{22}\right\rangle \right|+\\
    &\textstyle\phantom{{}\ge2(}\left|\left\langle B_{11}\right\rangle -\left\langle B_{21}\right\rangle \right|+\left|\left\langle B_{12}\right\rangle -\left\langle B_{22}\right\rangle \right|),
  \end{split}
  \label{eq:S_lower_from_ab}\\
  \Delta &\textstyle \le4-\left[-1+\frac{1}{2}s_{1}\left(\left\langle A_{11}B_{11}\right\rangle ,\left\langle A_{12}B_{12}\right\rangle ,\left\langle A_{21}B_{21}\right\rangle ,\left\langle A_{22}B_{22}\right\rangle \right)\right],\label{eq:S_upper_from_r}\\
  \begin{split}
    \Delta &\textstyle \le4-\frac{1}{2}(\left|\left\langle A_{11}\right\rangle +\left\langle A_{12}\right\rangle \right|+\left|\left\langle A_{21}\right\rangle +\left\langle A_{22}\right\rangle \right|+\\
    &\textstyle\phantom{{}\le4-2(}\left|\left\langle B_{11}\right\rangle +\left\langle B_{21}\right\rangle \right|+\left|\left\langle B_{12}\right\rangle +\left\langle B_{22}\right\rangle \right|),\label{eq:S_upper_from_ab}
  \end{split}
\end{align}
where $s_{1}(\cdots)$ is defined in \eqref{eq:s1} in Sec.~\ref{sec:techdetails} below and is equal to
the $\max\left|\ldots\right|$-part of \eqref{eq:C^0 general}. These
inequalities are always mutually compatible, whence $\Delta_{\min}$
is the larger of the two right-hand expressions in (\ref{eq:S_lower_from_r})
and (\ref{eq:S_lower_from_ab}). 
\end{proof}
It follows that $\Delta_{\min}-\Delta_{0}$ is always nonnegative,
and Definition \ref{def:contextuality general} is well-constructed:
$\Delta_{\min}-\Delta_{0}=0$ indicates no contextuality, $\Delta_{\min}-\Delta_{0}>0$
indicates contextuality on top of the direct cross-influences.

We can present the notion of (non-)contextuality in as close a form
as possible to the traditional CHSH inequalities: 
\begin{thm}
The system exhibits no contextuality if and only if 
\begin{equation}
\begin{split}
\left|\left\langle A_{11}B_{11}\right\rangle +\left\langle A_{12}B_{12}\right\rangle +\left\langle A_{21}B_{21}\right\rangle -\left\langle A_{22}B_{22}\right\rangle \right|&\leq2\left(1+\Delta_{0}\right),\\
\left|\left\langle A_{11}B_{11}\right\rangle +\left\langle A_{12}B_{12}\right\rangle -\left\langle A_{21}B_{21}\right\rangle +\left\langle A_{22}B_{22}\right\rangle \right|&\leq2\left(1+\Delta_{0}\right),\\
\left|\left\langle A_{11}B_{11}\right\rangle -\left\langle A_{12}B_{12}\right\rangle +\left\langle A_{21}B_{21}\right\rangle +\left\langle A_{22}B_{22}\right\rangle \right|&\leq2\left(1+\Delta_{0}\right),\\
\left|-\left\langle A_{11}B_{11}\right\rangle +\left\langle A_{12}B_{12}\right\rangle +\left\langle A_{21}B_{21}\right\rangle +\left\langle A_{22}B_{22}\right\rangle \right|&\leq2\left(1+\Delta_{0}\right),
\end{split}\label{eq:familar}
\end{equation}
where $\Delta_{0}$ is the natural measure of violation of marginal
selectivity, (\ref{eq:C_0 general}). If at least one of these inequalities
is violated, then the largest difference between the left-hand side
and $2\left(1+\Delta_{0}\right)$ is the degree of contextuality (after
scaling by $\nicefrac{1}{2}$).
\end{thm}
The maximum value attainable by one of the linear combinations in
(\ref{eq:familar}) is 4. It follows that the system exhibits no contextuality
if the violation of marginal selectivity $\Delta_{0}$ in it is not
less than 1. Put differently, if $\Delta_{0}\geq1$, any observed
distributions of random variables can be accounted for in terms of
direct cross-influences, with no contextuality involved.

\section{\label{sec:Consequences-of-the}Consequences of the new definition
of contextuality}

The notion of contextuality was presented in Introduction to mean
that random variables recorded under mutually incompatible conditions
cannot be ``sewn together'' into a single system of jointly distributed
random variables, provided one assumes that all or some of them preserve
their identity across different conditions. We should now relax the
assumption clause: 
\begin{quote}
contextuality means that random variables recorded under mutually
incompatible conditions cannot be ``sewn together'' into a single
system of jointly distributed random variables, provided one assumes
that their identity across different conditions changes as little
as possibly allowed by direct cross-influences (equivalently, by observed
deviations from marginal selectivity). 
\end{quote}
As mentioned in Introduction, marginal selectivity is rarely satisfied
outside quantum physics, and, in particular, is almost always violated
in psychological experiments. Consider, e.g., a double-detection experiment,
where a participant is presented two side-by-side flashes of light
(left and right) and asked to say ``Yes/No'' to the question ``Is
there a flash on the left?'' and another ``Yes/No'' to the question
``Is there a flash on the right?''. Each flash can be presented
at two intensity levels: zero (no flash) and some very small value
$s>0$. We have therefore four conditions: $\left(0,0\right),\left(0,s\right),\left(s,0\right),$$\left(s,s\right)$.
Denoting the response about the left stimulus by $A$ and he response
about the right stimulus by $B$, we get the eight random variables
$A_{00},B_{00},\ldots,A_{ss},B_{ss}$. The situation is formally identical
to the Alice-Bob paradigm. The ``normative'' diagram (\ref{eq:diagram selective}),
with $\alpha,\beta$ being the two flash intensities, is very likely
to be violated on the level of marginal probabilities: the answer
about the left flash will almost certainly be influenced by the intensity
of the right flash, and vice versa. Our definition of contextuality,
however, allows one to determine whether contextuality is there on
top of these direct cross-influences.

Another example is taken from the work by \citet{AertsGaboraSozzo2013}.
They estimated the probabilities with which people chose one of two
presented to them animal names and one of two presented to them animal
sounds. The results were as follows:\medskip{}

\fbox{\begin{minipage}[t]{0.9\columnwidth}%
\begin{center}
  Probability estimates from Table 1 of \citep{AertsGaboraSozzo2013}.$^{\;\dagger}$
\end{center}
\medskip

\begin{center}
\begin{small}
\setlength{\tabcolsep}{2pt}
\begin{tabular}{c|cc|cc@{\hspace{6pt}}c|cc|c}
\cline{2-3} \cline{7-8} 
\multirow{2}{*}{$\phi=(\alpha_{1},\beta_{1})$} & $B_{11}=$ & $B_{11}=$ &  &  & \multirow{2}{*}{$\phi=(\alpha_{1},\beta_{2})$} & $B_{12}=$ & $B_{12}=$ & \tabularnewline
 & Growls & Whinnies &  &  &  & Snorts & Meows & \tabularnewline
\cline{1-4} \cline{6-9} 
\multicolumn{1}{|c|}{$A_{11}=\textnormal{Horse}$} & .049 & .630 & \multicolumn{1}{c|}{.679} & \multicolumn{1}{c|}{} & $A_{12}=\textnormal{Horse}$ & .593 & .025 & \multicolumn{1}{c|}{.618}\tabularnewline
\multicolumn{1}{|c|}{$A_{11}=\textnormal{Bear}$} & .259 & .062 & \multicolumn{1}{c|}{.321} & \multicolumn{1}{c|}{} & $A_{12}=\textnormal{Bear}$ & .296 & .086 & \multicolumn{1}{c|}{.382}\tabularnewline
\cline{1-4} \cline{6-9} 
 & .308 & .692 &  &  &  & .889 & .111 & \tabularnewline
\cline{2-3} \cline{7-8} 
\multicolumn{1}{c}{} &  & \multicolumn{1}{c}{} &  &  & \multicolumn{1}{c}{} &  & \multicolumn{1}{c}{} & \tabularnewline
\cline{2-3} \cline{7-8} 
\multirow{2}{*}{$\phi=(\alpha_{2},\beta_{1})$} & $B_{21}=$ & $B_{21}=$ &  &  & \multirow{2}{*}{$\phi=(\alpha_{2},\beta_{2})$} & $B_{22}=$ & $B_{22}=$ & \tabularnewline
 & Growls & Whinnies &  &  &  & Snorts & Meows & \tabularnewline
\cline{1-4} \cline{6-9} 
\multicolumn{1}{|c|}{$A_{21}=\textnormal{Tiger}$} & .778 & .086 & \multicolumn{1}{c|}{.864} & \multicolumn{1}{c|}{} & $A_{22}=\textnormal{Tiger}$ & .148 & .086 & \multicolumn{1}{c|}{.234}\tabularnewline
\multicolumn{1}{|c|}{$A_{21}=\textnormal{Cat}$} & .086 & .049 & \multicolumn{1}{c|}{.135} & \multicolumn{1}{c|}{} & $A_{22}=\textnormal{Cat}$ & .099 & .667 & \multicolumn{1}{c|}{.766}\tabularnewline
\cline{1-4} \cline{6-9} 
 & .864 & .135 &  &  &  & .247 & .753 & \tabularnewline
\cline{2-3} \cline{7-8} 
\end{tabular}
\setlength{\tabcolsep}{6pt}
\end{small}
\par
\end{center}
\medskip{}

$^{\dagger\;}$Based on 81 respondents per table.

\medskip{}
\end{minipage}}

\medskip{}

Here, $\alpha$ indicates one of the two animal dichotomies offered
($\alpha_{1}=\textnormal{Horse or Bear}$, $\alpha_{2}=\textnormal{Tiger or Cat}$),
and $\beta$ analogously indicates one of two animal sound dichotomies.
The value of $\Delta_{\textnormal{CHSH}}$ given by (\ref{eq:C^0 general})
equals 0.210 here, and Aerts et al.\ report it as evidence in favor
of contextuality (note that the CHSH bound of $2$ corresponds to
$\Delta_{\textnormal{CHSH}}=0$). We criticized this conclusion \citep{DzhafarovKujala2013Topics}
by pointing out that the derivation of the CHSH inequalities is not
valid without marginal selectivity, and the latter is clearly violated
in the data: e.g., $\Pr\left[B_{12}=\textnormal{Snorts}\right]=0.889$
while $\Pr\left[B_{22}=\textnormal{Snorts}\right]=0.247$.

We can now amend our criticism: the computation of $\Delta_{\textnormal{CHSH}}$
is meaningful even if marginal selectivity is contravened. One has,
however, to compare $\Delta_{\textnormal{CHSH}}$ to $\Delta_{0}$
of (\ref{eq:C_0 general}) rather than to zero, and to compute $\max\left\{ \Delta_{0},\Delta_{\textnormal{CHSH}}\right\} -\Delta_{0}$
as the measure of contextuality. Unfortunately for the Aerts et al.'s
conclusions, $\Delta_{0}$ in their data is too large (1.889) to allow
for nonzero contextuality.

In quantum physics, the no-signaling condition (a special case of
marginal selectivity) can be ensured by separating the outputs from
the ``wrong'' inputs by space-like intervals. There are, however,
some indications that in the well-known experiments by \citet{Weihs1998},
where space-like separation is claimed to be the case, some violations
of marginal selectivity were observed \citep{AdenierKhrennikov2007}.
If so, and whatever the physical cause of these violations, our new
approach provides a way of testing whether contextuality is still
present in the data.

Signaling is natural to assume in Leggett--Garg \citeyearpar{LeggettGarg1985}
-type systems, with three binary random variables $X,Y,Z$ tied to
three successive moments of time, $t_{1}<t_{2}<t_{3}$. Any two of
these three random variables can be measured together, in one experiment,
but not all three of them. If $X$ and $Z$ are measured together,
then (in accordance with our general approach, see \citealp{DK2014LNCSQualified,DK2014PLOSconditionalization,DK2014FOOP,DK2014Advances,DK2014Scripta})
the identity of $X$ as a random variable may be different from the
identity of $X$ when measured together with $Y$. This means that
$X$ in the two situations should be labelled differently, say, $X_{13}$
and $X_{12}$, respectively (based on the time moments involved).
Analogously, we have $Y_{12}$ and $Y_{23}$ depending on whether
$Y$ is measured together with $X$ or with $Z$; and we have $Z_{13}$
and $Z_{23}$.

\citet{SuppesZanotti1981} have shown that given uniform marginals,
an equivalent condition for the existence of a joint distribution
of 
\begin{equation}
X_{12},X_{13},Y_{12},Y_{23},Z_{13},Z_{23}\label{eq:LG-obs-vars}
\end{equation}
under the constraint $X_{12}=X_{13}$, $Y_{12}=Y_{23}$, $Z_{13}=Z_{23}$
is
\begin{equation}
  \begin{split}
    -1&\le\left\langle X_{12}Y_{12}\right\rangle +\left\langle Y_{23}Z_{23}\right\rangle +\left\langle X_{13}Z_{13}\right\rangle\\
    &\leq1+2\max\left\{ \left\langle X_{12}Y_{12}\right\rangle ,\left\langle Y_{23}Z_{23}\right\rangle ,\left\langle X_{13}Z_{13}\right\rangle \right\} .\label{eq:LG-suppes-zanotti}
  \end{split}
\end{equation}
As a side product of our analysis, we show that this inequality in
fact holds for arbitrary marginals as well and we generalize the inequalities
to the signaling case.
\begin{thm}
The minimum possible value $\Delta'_{\min}$ for 
\begin{equation}
\Delta'=\Pr\left[X_{12}\ne X_{13}\right]+\Pr\left[Y_{12}\ne Y_{23}\right]+\Pr\left[Z_{13}\ne Z_{23}\right]\label{eq:LG-Delta}
\end{equation}
that is compatible with the observed expectations 
\begin{equation}
\left\langle X_{12}Y_{12}\right\rangle ,\,\left\langle X_{13}Z_{13}\right\rangle ,\,\left\langle Y_{23}Z_{23}\right\rangle ,\,\left\langle X_{12}\right\rangle ,\,\left\langle X_{13}\right\rangle ,\,\left\langle Y_{12}\right\rangle ,\,\left\langle Y_{23}\right\rangle ,\,\left\langle Z_{13}\right\rangle ,\,\left\langle Z_{23}\right\rangle \label{eq:LG-observable-expectations}
\end{equation}
 is
\begin{equation}
\Delta'_{\min}=\max\left\{ \Delta'_{0},\Delta'_{\textnormal{SZ}}\right\} ,
\end{equation}
where
\begin{equation}
\Delta'_{0}=\frac{1}{2}\left(\left|\left\langle X_{12}\right\rangle -\left\langle X_{13}\right\rangle \right|+\left|\left\langle Y_{12}\right\rangle -\left\langle Y_{23}\right\rangle \right|+\left|\left\langle Z_{13}\right\rangle -\left\langle Z_{23}\right\rangle \right|\right)
\end{equation}
is the natural measure of the violation of marginal selectivity and
\begin{equation}
\begin{array}{r@{}l}
\Delta'_{\textnormal{SZ}}=-\frac{1}{2}+\frac{1}{2}\max\big\{\, & \left\langle X_{12}Y_{12}\right\rangle +\left\langle X_{13}Z_{13}\right\rangle -\left\langle Y_{23}Z_{23}\right\rangle ,\\
 & \left\langle X_{12}Y_{12}\right\rangle -\left\langle X_{13}Z_{13}\right\rangle +\left\langle Y_{23}Z_{23}\right\rangle ,\\
 -& \left\langle X_{12}Y_{12}\right\rangle +\left\langle X_{13}Z_{13}\right\rangle +\left\langle Y_{23}Z_{23}\right\rangle ,\\
 -& \left\langle X_{12}Y_{12}\right\rangle -\left\langle X_{13}Z_{13}\right\rangle -\left\langle Y_{23}Z_{23}\right\rangle \big\}
\end{array}
\end{equation}
is ($\nicefrac{1}{2}$ times) the maximum violation of the Suppes--Zanotti
inequalities (\ref{eq:LG-suppes-zanotti}).\end{thm}
\begin{proof}
By Lemma~\ref{lem:2-1} of the next section, $\Delta'$ is compatible
with the observed expectations (\ref{eq:LG-observable-expectations})
if and only if it satisfies
\begin{align}
\Delta' &\textstyle \ge-\frac{1}{2}+\frac{1}{2}s_{1}\left(\left\langle X_{12}Y_{12}\right\rangle ,\left\langle Y_{23}Z_{23}\right\rangle ,\left\langle X_{13}Z_{13}\right\rangle \right),\label{eq:S_lower_from_rxyz}\\
\Delta' &\textstyle \ge\frac{1}{2}\left(\left|\left\langle X_{12}\right\rangle -\left\langle X_{13}\right\rangle \right|+\left|\left\langle Y_{12}\right\rangle -\left\langle Y_{23}\right\rangle \right|+\left|\left\langle Z_{13}\right\rangle -\left\langle Z_{23}\right\rangle \right|\right),\label{eq:S_lower_from_xyz}\\
\Delta' &\textstyle \le3-\left[-\frac{1}{2}-\frac{1}{2}s_{1}\left(\left\langle X_{12}Y_{12}\right\rangle ,\left\langle Y_{23}Z_{23}\right\rangle ,\left\langle X_{13}Z_{13}\right\rangle \right)\right],\label{eq:S_upper_from_rxyz}\\
\Delta' &\textstyle \le3-\frac{1}{2}\left(\left|\left\langle X_{12}\right\rangle +\left\langle X_{13}\right\rangle \right|+\left|\left\langle Y_{12}\right\rangle +\left\langle Y_{23}\right\rangle \right|+\left|\left\langle Z_{13}\right\rangle +\left\langle Z_{23}\right\rangle \right|\right).\label{eq:S_upper_from_xyz}
\end{align}
These inequalities are always mutually compatible, whence $\Delta'_{\min}$
is the larger of the two right-hand expressions in (\ref{eq:S_lower_from_rxyz})
and (\ref{eq:S_lower_from_xyz}). \end{proof}
\begin{defn}
The degree of contextuality in a system with given observed expectations
(\ref{eq:LG-observable-expectations}) is $\Delta'_{\min}-\Delta'_{0}$,
where $\Delta'_{\min}$ is the minimal value of $\Delta'$ in (\ref{eq:LG-Delta})
for which a joint distribution for (\ref{eq:LG-obs-vars}) exists. 
\end{defn}

Using essentially the same reasoning as for the EPR/Bohm paradigm,
we come to the following 
\begin{thm}
A Leggett--Garg-type systems exhibits no contextuality if and only
if
\begin{equation}
\begin{split}
\left\langle X_{12}Y_{12}\right\rangle +\left\langle Y_{23}Z_{23}\right\rangle -\left\langle X_{13}Z_{13}\right\rangle &\leq1+2\Delta'_{0},\\
\left\langle X_{12}Y_{12}\right\rangle -\left\langle Y_{23}Z_{23}\right\rangle +\left\langle X_{13}Z_{13}\right\rangle &\leq1+2\Delta'_{0},\\
-\left\langle X_{12}Y_{12}\right\rangle +\left\langle Y_{23}Z_{23}\right\rangle +\left\langle X_{13}Z_{13}\right\rangle &\leq1+2\Delta'_{0},\\
-\left\langle X_{12}Y_{12}\right\rangle -\left\langle Y_{23}Z_{23}\right\rangle -\left\langle X_{13}Z_{13}\right\rangle &\leq1+2\Delta'_{0}.
\end{split}\label{eq:Leggett-Garg-noncontextuality}
\end{equation}
The largest in absolute value breach of one of these boundaries then
can be taken as a measure of contextuality.
\end{thm}
Inequalities (\ref{eq:Leggett-Garg-noncontextuality}) can also be
equivalently rewritten closer to the Suppes--Zanotti \citeyearpar{SuppesZanotti1981}
formulation:
\begin{equation}
  \begin{split}
    -1-2\Delta'_{0}&\le\left\langle X_{12}Y_{12}\right\rangle +\left\langle Y_{23}Z_{23}\right\rangle +\left\langle X_{13}Z_{13}\right\rangle\\
    &\leq1+2\Delta'_{0}+2\max\left\{ \left\langle X_{12}Y_{12}\right\rangle ,\left\langle Y_{23}Z_{23}\right\rangle ,\left\langle X_{13}Z_{13}\right\rangle \right\} .\label{eq:Leggett-Garg-non-contextuality-Suppes-Zanotti}
  \end{split}
\end{equation}

\section{Technical details}\label{sec:techdetails}

In this section, we give the technical details of the computer-assisted
results used above. Refer to Fig.~\ref{fig:A-representation-of}
for a graphical representation of the connections and observed pairs
of random variables in the system.\footnote{Based on our most recent theoretical results \citep{KujalaDzhafarovLarsson2015,KujalaDzhafarov2015},
the computer-assisted proofs for the systems considered here can in
fact be obtained analytically as well. However, the principles of
computer-assisted proof laid out here are applicable in systems that
are not covered by the analytical results.}
\begin{lem}
\label{lem:1}The necessary and sufficient condition for the connection
expectations $\left(\left\langle A_{i1}A_{i2}\right\rangle ,\left\langle B_{1j}B_{2j}\right\rangle \right)_{i,j\in\left\{ 1,2\right\} }$
to be compatible with the observed expectations 
$$\left(\left\langle A_{ij}B_{ij}\right\rangle ,\left\langle A_{ij}\right\rangle ,\left\langle B_{ij}\right\rangle \right)_{i,j\in\left\{ 1,2\right\} }$$
is 
\begin{equation}
\begin{split}
  s_{0}&\left(\left\langle A_{11}B_{11}\right\rangle ,\left\langle A_{12}B_{12}\right\rangle ,\left\langle A_{21}B_{21}\right\rangle ,\left\langle A_{22}B_{22}\right\rangle \right)\\
  &\le6-s_{1}\left(\left\langle A_{11}A_{12}\right\rangle ,\left\langle B_{11}B_{21}\right\rangle ,\left\langle A_{21}A_{22}\right\rangle ,\left\langle B_{12}B_{22}\right\rangle \right),\\
  s_{1}&\left(\left\langle A_{11}B_{11}\right\rangle ,\left\langle A_{12}B_{12}\right\rangle ,\left\langle A_{21}B_{21}\right\rangle ,\left\langle A_{22}B_{22}\right\rangle \right)\\
  &\le6-s_{0}\left(\left\langle A_{11}A_{12}\right\rangle ,\left\langle B_{11}B_{21}\right\rangle ,\left\langle A_{21}A_{22}\right\rangle ,\left\langle B_{12}B_{22}\right\rangle \right),
\end{split}\label{eq:compatibility}
\end{equation}
where
\begin{equation}
\begin{split}
  s_{0}\left(a,b,c,d\right) & =\max\left\{ \left(\pm a\pm b\pm c\pm d\right):\textnormal{ the number of minuses is even}\right\},\\
  s_{1}\left(a,b,c,d\right) & =\max\left\{ \left(\pm a\pm b\pm c\pm d\right):\textnormal{ the number of minuses is odd}\right\}.\label{eq:s1}
\end{split}
\end{equation}
\end{lem}
\begin{proof}
The joint distribution of the eight random variables 
$$A_{11},B_{11},A_{12},B_{12},A_{21},B_{21},A_{22},B_{22}$$
is fully described by the vector $\mathbf{q}\in[0,1]^{n},$ $q_{1}+\dots+q_{n}=1$,
consisting of the probabilities of the $n=2^{8}=256$ different combinations
of the values of the eight random variables. We then define a vector
$\mathbf{p}\in[0,1]^{m}$, $m=32$, consisting of the $16$ observable
probabilities $\Pr[A_{ij}=a,\ B_{ij}=b]$ for $a,b\in\{-1,1\},$ $i,j\in\{1,2\}$
and the $16$ connection probabilities given by $\Pr[A_{i1}=a,\ A_{i2}=a']$
and $\Pr[B_{1j}=b,\ B_{2j}=b']$ for $a,a',b,b'\in\{-1,1\}$ and $i,j\in\{1,2\}$.
As every element of $\mathbf{p}$ is a ($2$-)marginal probability
of the joint represented by $\mathbf{q}$, there exists a binary marix
$M\in\{0,1\}^{m\times n}$ such that
\begin{equation}
\mathbf{p}=M\mathbf{q}.\label{eq:pMq}
\end{equation}
It follows that the observable probabilities $p_{1},\dots,p_{16}$
are compatible with the connection probabilities $p_{17},\dots,p_{32}$
if and only if there exists an $n$-vector $\mathbf{q}\ge0$ such
that \eqref{eq:pMq} holds. As described in \citep[Text~S3]{DK2013PLoS},
the set of vectors $\mathbf{p}$ satisfying this constraint
forms a polytope whose vertices are given by the columns of $M$ and
whose half space representation can be obtained by a facet enumeration
algorithm. As also described in \citep{DK2013PLoS},
this halfspace representation consists of $160$ inequalities and
$16$ equations in $p_{1},\dots,p_{32}$. The $16$ equations correspond
to the requirement that the $1$-marginals of the observable probabilities
agree with those of the connections and that the observable probabilities
are properly normalized.

Expressing the probabilities in the vector $\mathbf{p}$ in terms
of the observable and connection expectations $\left(\left\langle A_{ij}B_{ij}\right\rangle ,\left\langle A_{ij}\right\rangle ,\left\langle B_{ij}\right\rangle ,\left\langle A_{i1}A_{i2}\right\rangle ,\left\langle B_{1j}B_{2j}\right\rangle \right)$,
$i,j\in\{1,2\}$, the $16$ equations become identically true (the
parameterization already guarantees them), and of the $160$ inequalities,
$128$ turn into exactly those represented by \eqref{eq:compatibility}
and the remaining $32$ are trivial constraints of the form 
\begin{equation}
-1+|\left\langle A\right\rangle +\left\langle B\right\rangle |\le\left\langle AB\right\rangle \le1-|\left\langle A\right\rangle -\left\langle B\right\rangle |\label{eq:implicit}
\end{equation}
 for the $8$ pairs of random variables involved in \eqref{eq:compatibility}.
The trivial constraints correspond to the implicit requirement that
the observable and connection probabilities are nonnegative and thus
they need not be explicitly shown in the statement of the theorem.
\end{proof}
This proof is different from the similar result in \citep{DK2013PLoS}
in that the parameterization for the probabilities in $\mathbf{p}$
is more general (allowing for arbitrary marginals of the eight random
variables) and so we obtain a more general condition for the compatibility
of observable and connection probabilities than before. It should
be noted that although the expectations $\left\langle A_{ij}\right\rangle ,\left\langle B_{ij}\right\rangle $,
$i,j\in\{1,2\}$ do not explicitly appear in \eqref{eq:compatibility},
they are still present in the $32$ implicit constraints.
\begin{lem}
\label{lem:2}If the connection expectations $\left(\left\langle A_{i1}A_{i2}\right\rangle ,\left\langle B_{1j}B_{2j}\right\rangle \right)_{i,j\in\left\{ 1,2\right\} }$
are compatible with the observed expectations $\left(\left\langle A_{ij}B_{ij}\right\rangle ,\left\langle A_{ij}\right\rangle ,\left\langle B_{ij}\right\rangle \right)_{i,j\in\left\{ 1,2\right\} }$,
then, with $\Delta$ defined as in \eqref{eq:C under MS}, 
\begin{equation}
\begin{split}
\Delta&\textstyle\ge-1+\frac{1}{2}s_{1}\left(\left\langle A_{11}B_{11}\right\rangle ,\left\langle A_{12}B_{12}\right\rangle ,\left\langle A_{21}B_{21}\right\rangle ,\left\langle A_{22}B_{22}\right\rangle \right),\\
\Delta&\textstyle\ge\frac{1}{2}(\left|\left\langle A_{11}\right\rangle -\left\langle A_{12}\right\rangle \right|+\left|\left\langle A_{21}\right\rangle -\left\langle A_{22}\right\rangle \right|+\\
&\phantom{{}\ge2(}\left|\left\langle B_{11}\right\rangle -\left\langle B_{21}\right\rangle \right|+\left|\left\langle B_{12}\right\rangle -\left\langle B_{22}\right\rangle \right|),\\
\Delta&\textstyle\le4-\left[-1+\frac{1}{2}s_{1}\left(\left\langle A_{11}B_{11}\right\rangle ,\left\langle A_{12}B_{12}\right\rangle ,\left\langle A_{21}B_{21}\right\rangle ,\left\langle A_{22}B_{22}\right\rangle \right)\right],\\
\Delta&\textstyle\le4-\frac{1}{2}(\left|\left\langle A_{11}\right\rangle +\left\langle A_{12}\right\rangle \right|+\left|\left\langle A_{21}\right\rangle +\left\langle A_{22}\right\rangle \right|+\\
&\phantom{{}\le4-2(}\left|\left\langle B_{11}\right\rangle +\left\langle B_{21}\right\rangle \right|+\left|\left\langle B_{12}\right\rangle +\left\langle B_{22}\right\rangle \right|).
\end{split}\label{eq:Delta_system}
\end{equation}
Conversely, if these inequalities are satisfied for a given value
of $\Delta$, then the connection expectations $\left(\left\langle A_{i1}A_{i2}\right\rangle ,\left\langle B_{1j}B_{2j}\right\rangle \right)_{i,j\in\left\{ 1,2\right\} }$
can always be chosen so that they are compatible with the observable
expectations $\left(\left\langle A_{ij}B_{ij}\right\rangle ,\left\langle A_{ij}\right\rangle ,\left\langle B_{ij}\right\rangle \right)_{i,j\in\left\{ 1,2\right\} }$
and yield the given value of $\Delta$ in \eqref{eq:C under MS}.\end{lem}
\begin{proof}
Given the 160 inequalities (including the 32 implicit inequalities)
of Lemma~\ref{lem:1} characterizing the compatibility of the connection
expectations with the observable expectations, we amend this linear
system with the equation \eqref{eq:C under MS} defining $\Delta$
written in terms of the expectations 
$$\left(\left\langle A_{i1}A_{i2}\right\rangle ,\left\langle B_{1j}B_{2j}\right\rangle ,\left\langle A_{ij}\right\rangle ,\left\langle B_{ij}\right\rangle \right)_{i,j\in\{1,2\}}.$$
Then, we use this equation to eliminate one of the connection expectation
variables $\left(\left\langle A_{i1}A_{i2}\right\rangle ,\left\langle B_{1j}B_{2j}\right\rangle \right)_{i,j\in\{1,2\}}$
from the system (by solving the variable from the equation and then
substituting the solution everywhere else). After that, we eliminate
the three remaining connection expectation variables one by one using
the Fourier--Motzkin elimination algorithm (see Theorem~\ref{thm:Fourier-Motzkin}
below). After the elimination of each variable, we remove any redundant
inequalities from the system by linear programming using the algorithm
described in \citep[Text S3]{DK2013PLoS}. After having eliminated
all connection expectation variables, we are left with the system
\eqref{eq:Delta_system} (and implicit constraints of the form \eqref{eq:implicit}
for the pairs $(A_{ij},B_{ij})$, $i,j\in\{1,2\}$). The Fourier--Motzkin
elimination algorithm guarantees that the resulting system has a solution
precisely when the original system has a solution with \emph{some}
values of the eliminated variables.\end{proof}
\begin{thm}
[Fourier--Motzkin elimination]\label{thm:Fourier-Motzkin}\index{Fourier--Motzkin elimination}Given a
system of linear inequalities in the variables $x$ and $\mathbf{y=}y_{1},\dots,y_{n}$,
the system can always be rearranged in the following form
\begin{equation}
\begin{array}{lll}
x\ge\mathbf{l}_{i}\mathbf{\cdot y}, &  & i=1,\dots,n_{\mathbf{l}},\\
x\le\mathbf{u}_{i}\cdot\mathbf{y}, &  & i=1,\dots,n_{\mathbf{u}},\\
0\le\mathbf{n}_{i}\mathbf{\cdot y}, &  & i=1,\dots,n_{\mathbf{n}},
\end{array}
\end{equation}
where $\mathbf{l}_{1},\dots,\mathbf{l}_{n_{l}},\mathbf{u}_{1},\dots,\mathbf{\mathbf{u}}_{n_{u}},\mathbf{n}_{1},\dots,\mathbf{n}_{n_{n}}\in\mathbb{R}^{n}$.
Furthermore, given $\mathbf{y}\in\mathbb{R}$, this system is solved
by $\mathbf{y}$ and some $x\in\mathbb{R}$ if and only if the following
system is solved by $\mathbf{y}$: 
\begin{equation}
\begin{array}{rcl}
\mathbf{l}_{i}\cdot\mathbf{y}\le\mathbf{u}_{j}\cdot\mathbf{y}, &  & i=1,\dots,n_{\mathbf{l}},\ j=1,\dots,n_{\mathbf{u}},\\
0\le\mathbf{n}_{i}\mathbf{\cdot y}, &  & i=1,\dots,n_{\mathbf{n}}.
\end{array}
\end{equation}

\end{thm}
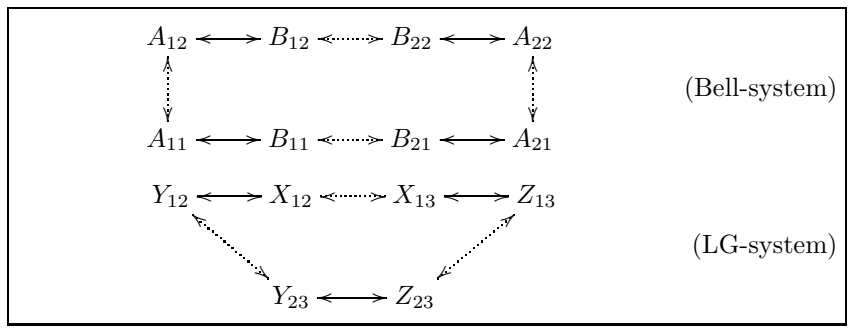
\begin{figure}
\begin{center}
  \fbox{\begin{minipage}[c]{.9\textwidth}%
      \[
      \begin{array}{c}
        \xymatrix{A_{12}\ar@{.>}[d]\ar[r] & B_{12}\ar@{.>}[r]\ar[l] & B_{22}\ar@{.>}[l]\ar[r] & A_{22}\ar@{.>}[d]\ar[l]\\
          A_{11}\ar@{.>}[u]\ar[r] & B_{11}\ar@{.>}[r]\ar[l] & B_{21}\ar@{.>}[l]\ar[r] & A_{21}\ar@{.>}[u]\ar[l]
        }
      \end{array}\tag{Bell-system}
      \]
      \[
      \begin{array}{c}
        \xymatrix{Y_{12}\ar@{.>}[dr]\ar[r] & X_{12}\ar@{.>}[r]\ar[l] & X_{13}\ar@{.>}[l]\ar[r] & Z_{13}\ar@{.>}[dl]\ar[l]\\
          & Y_{23}\ar@{.>}[ul]\ar[r] & Z_{23}\ar@{.>}[ur]\ar[l]
        }
      \end{array}\tag{LG-system}
      \]
  \end{minipage}}
\end{center}

\protect\caption[.]{Random variables involved in the Bell-system and LG-system. The pairs
of random variables whose joint distributions are empirically observed,
e.g., $\left(A_{12},B_{12}\right)$ and $\left(X_{12},Y_{12}\right)$,
are indicated by solid double-arrows. The pairs of random variables
forming probabilistic connections (with unobservable joint distributions)
are indicated by point double-arrows, e.g., $\left(A_{11},A_{12}\right)$
and $\left(X_{12},X_{13}\right)$. \label{fig:A-representation-of} }
\end{figure}

\begin{lem}
\label{lem:1-1}The necessary and sufficient condition for the connection
expectations $\left\langle X_{12}X_{13}\right\rangle $, $\left\langle Y_{12}Y_{23}\right\rangle $,
$\left\langle Z_{13}Z_{23}\right\rangle $ to be compatible with the
observed expectations $\left\langle X_{12}Y_{12}\right\rangle $,
$\left\langle X_{13}Z_{13}\right\rangle $, $\left\langle Y_{23}Z_{23}\right\rangle $,
$\left\langle X_{12}\right\rangle $, $\left\langle X_{13}\right\rangle $,
$\left\langle Y_{12}\right\rangle $, $\left\langle Y_{23}\right\rangle $,
$\left\langle Z_{13}\right\rangle $, $\left\langle Z_{23}\right\rangle $
is 
\begin{equation}
s_{1}\left(\left\langle X_{12}Y_{12}\right\rangle ,\left\langle X_{13}Z_{13}\right\rangle ,\left\langle Y_{23}Z_{23}\right\rangle ,\left\langle X_{12}X_{13}\right\rangle ,\left\langle Y_{12}Y_{23}\right\rangle ,\left\langle Z_{13}Z_{23}\right\rangle \right)\le4.\label{eq:compatibility-1}
\end{equation}
where
\begin{equation}
\begin{split}
  s_{1}\left(a,b,c,d,e,f\right) = \max \{ &\left(\pm a\pm b\pm c\pm d\pm e\pm f\right):\\
  &\textnormal{ the number of minuses is odd}\,\} .\label{eq:s1-1}
\end{split}
\end{equation}
\end{lem}
\begin{proof}
The details are analogous to those of the proof of Lemma~\ref{lem:1}.
The polytope in terms of probabilities is defined by 12 equations
and 56 inequalities. The 12 equations correspond to the requirement
that the 1-marginals of the observable probabilities agree with those
of the connections and that the observable probabilities are properly
normalized. Expressing the probabilities in terms of the observable
and connection expections, the 16 equations become identically true
and of the 56 inequalities, 32 turn into those represented by (\ref{eq:s1-1})
and the remaining 24 correspond to the trivial constraints of the
form (\ref{eq:implicit}) for the 6 pairs of random variables appearing
in (\ref{eq:s1-1}).\end{proof}
\begin{lem}
\label{lem:2-1}If the connection expectations $\left\langle X_{12}X_{13}\right\rangle ,\left\langle Y_{12}Y_{23}\right\rangle ,\left\langle Z_{13}Z_{23}\right\rangle $
are compatible with the observed expectations 
$$\left\langle X_{12}Y_{12}\right\rangle ,\left\langle X_{13}Z_{13}\right\rangle ,\left\langle Y_{23}Z_{23}\right\rangle, \left\langle X_{12}\right\rangle ,\left\langle X_{13}\right\rangle ,\left\langle Y_{12}\right\rangle ,\left\langle Y_{23}\right\rangle ,\left\langle Z_{13}\right\rangle ,\left\langle Z_{23}\right\rangle ,$$
then, with $\Delta'$ defined as in \eqref{eq:LG-Delta}, 
\begin{equation}
\begin{split}
\Delta'&\textstyle\ge-\frac{1}{2}+\frac{1}{2}s_{1}\left(\left\langle X_{12}Y_{12}\right\rangle ,\left\langle X_{13}Z_{13}\right\rangle ,\left\langle Y_{23}Z_{23}\right\rangle \right),\\
\Delta'&\textstyle\ge\frac{1}{2}\left(\left|\left\langle X_{12}\right\rangle -\left\langle X_{13}\right\rangle \right|+\left|\left\langle Y_{12}\right\rangle -\left\langle Y_{23}\right\rangle \right|+\left|\left\langle Z_{13}\right\rangle -\left\langle Z_{23}\right\rangle \right|\right),\\
\Delta'&\textstyle\le3-\left[-\frac{1}{2}+\frac{1}{2}s_{0}\left(\left\langle X_{12}Y_{12}\right\rangle ,\left\langle X_{13}Z_{13}\right\rangle ,\left\langle Y_{23}Z_{23}\right\rangle \right)\right],\\
\Delta'&\textstyle\le3-\frac{1}{2}\left(\left|\left\langle X_{12}\right\rangle +\left\langle X_{13}\right\rangle \right|+\left|\left\langle Y_{12}\right\rangle +\left\langle Y_{23}\right\rangle \right|+\left|\left\langle Z_{13}\right\rangle +\left\langle Z_{23}\right\rangle \right|\right).
\end{split}\label{eq:Delta_system-1}
\end{equation}
Conversely, if these inequalities are satisfied for a given value
of $\Delta'$, then the connection expectations $\left\langle X_{12}X_{13}\right\rangle ,\left\langle Y_{12}Y_{23}\right\rangle ,\left\langle Z_{13}Z_{23}\right\rangle $
can always be chosen so that they are compatible with the observable
expectations 
$$\left\langle X_{12}Y_{12}\right\rangle ,\left\langle X_{13}Z_{13}\right\rangle ,\left\langle Y_{23}Z_{23}\right\rangle ,\left\langle X_{12}\right\rangle ,\left\langle X_{13}\right\rangle ,\left\langle Y_{12}\right\rangle ,\left\langle Y_{23}\right\rangle ,\left\langle Z_{13}\right\rangle ,\left\langle Z_{23}\right\rangle $$
and yield the given value of $\Delta'$ in \eqref{eq:LG-Delta}.\end{lem}
\begin{proof}
The details are analogous to those of the proof of Lemma~\eqref{lem:2}.
\end{proof}

\paragraph*{Acknowledgements. }

This work was supported by NSF grant SES-1155956 and AFOSR grant FA9550-14-1-0318. The authors are
grateful to J. Acacio de Barros and Gary Oas for numerous discussions
of issues related to contextuality. 

\bibliographystyle{apalike}
\bibliography{si.bib}

\begin{thebibliography}{44}
\newcommand{\enquote}[1]{#1}
\providecommand{\natexlab}[1]{#1}
\providecommand{\url}[1]{\texttt{#1}}
\providecommand{\urlprefix}{URL }
\expandafter\ifx\csname urlstyle\endcsname\relax
  \providecommand{\doi}[1]{doi:\discretionary{}{}{}#1}\else
  \providecommand{\doi}{doi:\discretionary{}{}{}\begingroup
  \urlstyle{rm}\Url}\fi

\bibitem[{Adenier and Khrennikov(2007)}]{AdenierKhrennikov2007}
Adenier, G. and Khrennikov, A.~{\relax Yu}. (2007). \enquote{Is the fair
  sampling assumption supported by {EPR} experiments?} \emph{Journal of Physics
  B} \textbf{40}, pp. 131--141.

\bibitem[{Aerts \emph{et~al.}(2013)Aerts, Gabora and
  Sozzo}]{AertsGaboraSozzo2013}
Aerts, D., Gabora, L. and Sozzo, S. (2013). \enquote{Concepts and their
  dynamics: A quantum theoretical model,} \emph{Topics in Cognitive Science}
  \textbf{5}, pp. 737--772.

\bibitem[{Aspect \emph{et~al.}(1981)Aspect, Grangier and
  Roger}]{AspectGrangierRoger1981}
Aspect, A., Grangier, P. and Roger, G. (1981). \enquote{Experimental tests of
  realistic local theories via {B}ell's theorem,} \emph{Physical Review
  Letters} \textbf{47}, pp. 460--463.

\bibitem[{Aspect \emph{et~al.}(1982)Aspect, Grangier and
  Roger}]{AspectGrangierRoger1982}
Aspect, A., Grangier, P. and Roger, G. (1982). \enquote{Experimental
  realization of {E}instein-{P}odolsky-{R}osen-{B}ohm gedankenexperiment: A new
  violation of bell's inequalities,} \emph{Physical Review Letters}
  \textbf{49}, pp. 91--94.

\bibitem[{Badzi{\c a}g \emph{et~al.}(2009)Badzi{\c a}g, Bengtsson, Cabello and
  Pitowsky}]{Badziag2009}
Badzi{\c a}g, P., Bengtsson, I., Cabello, A. and Pitowsky, I. (2009).
  \enquote{Universality of state-independent violation of correlation
  inequalities for noncontextual theories,} \emph{Physical Review Letters}
  \textbf{103}, p. 050401.

\bibitem[{Bell(1964)}]{Bell1964}
Bell, J.~S. (1964). \enquote{On the {E}instein {P}odolsky {R}osen paradox,}
  \emph{Physics} \textbf{1}, pp. 195--200.

\bibitem[{Cabello(2013)}]{Cabello2013PRL}
Cabello, A. (2013). \enquote{Simple explanation of the quantum violation of a
  fundamental inequality,} \emph{Physical Review Letters} \textbf{110}, p.
  060402.

\bibitem[{Cereceda(2000)}]{Cereceda2000}
Cereceda, J. (2000). \enquote{Quantum mechanical probabilities and general
  probabilistic constraints for {E}instein--{P}odolsky--{R}osen--{B}ohm
  experiments,} \emph{Foundations of Physics Letters} \textbf{13}, pp.
  427--442.

\bibitem[{Clauser \emph{et~al.}(1969)Clauser, Horne, Shimony and
  Holt}]{ClauserHorneShimonyHolt1969}
Clauser, J.~F., Horne, M.~A., Shimony, A. and Holt, R.~A. (1969).
  \enquote{Proposed experiment to test local hidden-variable theories,}
  \emph{Physical Review Letters} \textbf{23}, pp. 880--884.

\bibitem[{Dzhafarov(2003)}]{Dzhafarov2003c}
Dzhafarov, E.~N. (2003). \enquote{Selective influence through conditional
  independence,} \emph{Psychometrika} \textbf{68}, 1, pp. 7--25.

\bibitem[{Dzhafarov and Kujala(2010)}]{DzhafarovKujala2010}
Dzhafarov, E.~N. and Kujala, J.~V. (2010). \enquote{The joint distribution
  criterion and the distance tests for selective probabilistic causality,}
  \emph{Frontiers in Psychology} \textbf{1}, p. 151,
  \doi{10.3389/fpsyg.2010.00151}.

\bibitem[{Dzhafarov and Kujala(2012{\natexlab{a}})}]{DzhafarovKujala2012b}
Dzhafarov, E.~N. and Kujala, J.~V. (2012{\natexlab{a}}). \enquote{Quantum
  entanglement and the issue of selective influences in psychology: An
  overview,} in J.~R. Busemeyer, F.~Dubois, A.~Lambert-Mobiliansky and
  M.~Melucci. eds., \emph{Quantum Interaction}, \emph{Lecture Notes in Computer
  Science}, Vol. 7620, pp. 184--195, Springer.

\bibitem[{Dzhafarov and Kujala(2012{\natexlab{b}})}]{DzhafarovKujala2012a}
Dzhafarov, E.~N. and Kujala, J.~V. (2012{\natexlab{b}}). \enquote{Selectivity
  in probabilistic causality: Where psychology runs into quantum physics,}
  \emph{Journal of Mathematical Psychology} \textbf{56}, pp. 54--63.

\bibitem[{Dzhafarov and Kujala(2013{\natexlab{a}})}]{DK2013PLoS}
Dzhafarov, E.~N. and Kujala, J.~V. (2013{\natexlab{a}}).
  \enquote{All-possible-couplings approach to measuring probabilistic context,}
  \emph{PLoS ONE} \textbf{8}, 5, p. e61712, \doi{10.1371/journal.pone.0061712.}

\bibitem[{Dzhafarov and
  Kujala(2013{\natexlab{b}})}]{DzhafarovKujala2013ProcAMS}
Dzhafarov, E.~N. and Kujala, J.~V. (2013{\natexlab{b}}).
  \enquote{Order-distance and other metric-like functions on jointly
  distributed random variables,} \emph{Proceedings of the American Mathematical
  Society} \textbf{141}, pp. 3291--3301.

\bibitem[{Dzhafarov and Kujala(2014{\natexlab{a}})}]{DK2014Scripta}
Dzhafarov, E.~N. and Kujala, J.~V. (2014{\natexlab{a}}). \enquote{Contextuality
  is about identity of random variables.} \emph{Physica Scripta} \textbf{T163},
  p. 014009.

\bibitem[{Dzhafarov and
  Kujala(2014{\natexlab{b}})}]{DK2014PLOSconditionalization}
Dzhafarov, E.~N. and Kujala, J.~V. (2014{\natexlab{b}}). \enquote{Embedding
  quantum into classical: contextualization vs conditionalization,} \emph{PLoS
  ONE} \textbf{9}, 3, p. e92818.

\bibitem[{Dzhafarov and Kujala(2014{\natexlab{c}})}]{DK2014FOOP}
Dzhafarov, E.~N. and Kujala, J.~V. (2014{\natexlab{c}}). \enquote{No-forcing
  and no-matching theorems for classical probability applied to quantum
  mechanics,} \emph{Foundations of Physics} \textbf{44}, pp. 248--265.

\bibitem[{Dzhafarov and Kujala(2014{\natexlab{d}})}]{DzhafarovKujala2013Topics}
Dzhafarov, E.~N. and Kujala, J.~V. (2014{\natexlab{d}}). \enquote{On selective
  influences, marginal selectivity, and {B}ell/{CHSH} inequalities,}
  \emph{Topics in Cognitive Science} \textbf{6}, 1, pp. 121--128.

\bibitem[{Dzhafarov and Kujala(2014{\natexlab{e}})}]{DK2014LNCSQualified}
Dzhafarov, E.~N. and Kujala, J.~V. (2014{\natexlab{e}}). \enquote{A qualified
  {K}olmogorovian account of probabilistic contextuality,} \emph{Lecture Notes
  in Computer Science} \textbf{8369}, pp. 201--212.

\bibitem[{Dzhafarov and Kujala(2015)}]{DK2014Advances}
Dzhafarov, E.~N. and Kujala, J.~V. (2015). \enquote{Random variables recorded
  under mutually exclusive conditions: Contextuality-by-default,} in
  \emph{Advances in Cognitive Neurodynamics IV}, pp. 405--410.

\bibitem[{Dzhafarov and Kujala(forthcoming)}]{DzhafarovKujala_handbook}
Dzhafarov, E.~N. and Kujala, J.~V. (forthcoming). \enquote{Probability, random
  variables, and selectivity,} in W.~H. Batchelder, H.~Colonius, E.~Dzhafarov
  and J.~I. Myung. eds., \emph{The New Handbook of Mathematical Psychology}, to
  be published by Cambridge University Press.

\bibitem[{Dzhafarov \emph{et~al.}(2015)Dzhafarov, Kujala and
  Larsson}]{DKL2015FooP}
Dzhafarov, E.~N., Kujala, J.~V. and Larsson, J.-{\r A}. (2015).
  \enquote{Contextuality in three types of quantum-mechanical systems,}
  \emph{Foundations of Physics} \textbf{45}, pp. 762--782,
  \doi{10.1007/s10701-015-9882-9}.

\bibitem[{Fine(1982)}]{Fine1982b}
Fine, A. (1982). \enquote{Hidden variables, joint probability, and the bell
  inequalities,} \emph{Physical Review Letters} \textbf{48}, pp. 291--295.

\bibitem[{Khrennikov(2008{\natexlab{a}})}]{Khrennikov2008_Bell-Boole}
Khrennikov, A.~{\relax Yu}. (2008{\natexlab{a}}). \enquote{{B}ell-{B}oole
  inequality: Nonlocality or probabilistic incompatibility of random
  variables?} \emph{Entropy} \textbf{10}, pp. 19--32.

\bibitem[{Khrennikov(2008{\natexlab{b}})}]{Khrennikov2008_EPR-Bohm}
Khrennikov, A.~{\relax Yu}. (2008{\natexlab{b}}). \enquote{{EPR}-{B}ohm
  experiment and {B}ell's inequality: Quantum physics meets probability
  theory,} \emph{Theoretical and Mathematical Physics} \textbf{157}, pp.
  1448--1460.

\bibitem[{Khrennikov(2009)}]{Khrennikov2009}
Khrennikov, A.~{\relax Yu}. (2009). \emph{Contextual Approach to Quantum
  Formalism}, Fundamental Theories of Physics 160, Springer, Dordrecht.

\bibitem[{Kirchmair \emph{et~al.}(2009)Kirchmair, Z{\"a}hringer, Gerritsma,
  Kleinmann, G{\"u}hne, Cabello, Blatt and Roos}]{Kirchmair2009}
Kirchmair, G., Z{\"a}hringer, F., Gerritsma, R., Kleinmann, M., G{\"u}hne, O.,
  Cabello, A., Blatt, R. and Roos, C. (2009). \enquote{State-independent
  experimental test of quantum contextuality,} \emph{Nature} \textbf{460}, pp.
  494--497.

\bibitem[{Kochen and Specker(1967)}]{KochenSpecker1967}
Kochen, S. and Specker, F. (1967). \enquote{The problem of hidden variables in
  quantum mechanics,} \emph{Journal of Mathematics and Mechanics} \textbf{17},
  pp. 59--87.

\bibitem[{Kujala and Dzhafarov(2008)}]{KujalaDzhafarov2008b}
Kujala, J.~V. and Dzhafarov, E.~N. (2008). \enquote{Testing for selectivity in
  the dependence of random variables on external factors,} \emph{Journal of
  Mathematical Psychology} \textbf{52}, pp. 128--144.

\bibitem[{Kujala and Dzhafarov(2015)}]{KujalaDzhafarov2015}
Kujala, J.~V. and Dzhafarov, E.~N. (2015). \enquote{Proof of a conjecture on
  contextuality in cyclic systems with binary variables,}
  \emph{arXiv:1503.02181} .

\bibitem[{Kujala \emph{et~al.}(2015)Kujala, Dzhafarov and
  Larsson}]{KujalaDzhafarovLarsson2015}
Kujala, J.~V., Dzhafarov, E.~N. and Larsson, J.-{\AA}. (2015).
  \enquote{Necessary and sufficient conditions for maximal contextuality in a
  broad class of quantum mechanical systems,} \emph{arXiv:1412.4724} .

\bibitem[{Landau(1987)}]{Landau1987}
Landau, L.~J. (1987). \enquote{On the violation of {B}ell's inequality in
  quantum theory,} \emph{Physics Letters A} \textbf{120}, pp. 54--56.

\bibitem[{Larsson(2002)}]{Larsson2002}
Larsson, J.-{\AA}. (2002). \enquote{A {K}ochen-{S}pecker inequality,}
  \emph{Europhysics Letters} \textbf{58}, pp. 799--805.

\bibitem[{Laudisa(1997)}]{Laudisa1997}
Laudisa, F. (1997). \enquote{Contextualism and nonlocality in the algebra of
  {EPR} observables,} \emph{Philosophy of Science} \textbf{64}, pp. 478--496.

\bibitem[{Leggett and Garg(1985)}]{LeggettGarg1985}
Leggett, A.~J. and Garg, A. (1985). \enquote{Quantum mechanics versus
  macroscopic realism: is the flux there when nobody looks?} \emph{Physical
  Review Letters} \textbf{54}, pp. 857--860.

\bibitem[{Masanes \emph{et~al.}(2006)Masanes, Acin and Gisin}]{Masanes2006}
Masanes, {\relax Ll}., Acin, A. and Gisin, N. (2006). \enquote{General
  properties of nonsignaling theories,} \emph{Physical Review A} \textbf{73},
  p. 012112.

\bibitem[{Oas \emph{et~al.}(2014)Oas, de~Barros and Carvalhaes}]{Oas2014}
Oas, G., de~Barros, J.~A. and Carvalhaes, C. (2014). \enquote{Exploring
  non-signalling polytopes with negative probability,} \emph{Physica Scripta}
  \textbf{T163}, p. 014034.

\bibitem[{Specker(1960)}]{Specker1960}
Specker, E. (1960). \enquote{Die {L}ogik nicht gleichzeitig entscheidbarer
  {A}ussagen,} \emph{Dialectica} \textbf{14}, pp. 239--246, (English
  translation by M.P. Seevinck available as arXiv:1103.4537).

\bibitem[{Spekkens(2008)}]{Spekkens2008}
Spekkens, R.~W. (2008). \enquote{Negativity and contextuality are equivalent
  notions of nonclassicality,} \emph{Physical Review Letters} \textbf{101}, p.
  020401.

\bibitem[{Suppes and Zanotti(1981)}]{SuppesZanotti1981}
Suppes, P. and Zanotti, M. (1981). \enquote{When are probabilistic explanations
  possible?} \emph{Synthese} \textbf{48}, pp. 191--199.

\bibitem[{Townsend and Schweickert(1989)}]{TownsendSchweickert1989}
Townsend, J.~T. and Schweickert, R. (1989). \enquote{Toward the trichotomy
  method of reaction times: Laying the foundation of stochastic mental
  networks,} \emph{Journal of Mathematical Psychology} \textbf{33}, 3, pp. 309
  -- 327.

\bibitem[{Tsirelson(1980)}]{Tsirelson1980}
Tsirelson, B.~S. (1980). \enquote{Quantum generalizations of {B}ell's
  inequality,} \emph{Letters in Mathematical Physics} \textbf{4}, pp. 93--100.

\bibitem[{Weihs \emph{et~al.}(1998)Weihs, Jennewein, Simon, Weinfurter,  and
  Zeilinger}]{Weihs1998}
Weihs, G., Jennewein, T., Simon, C., Weinfurter, H.,  and Zeilinger, A. (1998).
  \enquote{Violation of {B}ell's inequality under strict {E}instein locality
  conditions,} \emph{Physical Review Letters} \textbf{81}, pp. 5039--5043.

\end{thebibliography}

\end{document}